\documentclass[english]{sbrt}
\usepackage[english]{babel}
\usepackage[utf8]{inputenc}

\usepackage{amssymb,amsmath,amsfonts,amsthm,bm,bbm}
\usepackage{amsthm}
\usepackage{mathrsfs}
\usepackage{mathtools}
\usepackage{algorithm}
\usepackage{algpseudocode} 
\usepackage{array}
\usepackage{xcolor}
\usepackage{color}
\usepackage{cite}
\usepackage{enumerate}

\usepackage{graphicx}
\usepackage{subcaption}
\newtheorem*{proposition}{Proposition}

\let\OLDthebibliography\thebibliography
\renewcommand\thebibliography[1]{
\OLDthebibliography{#1}
\setlength{\parskip}{0pt}
\setlength{\itemsep}{0pt plus 0.32ex}}
\usepackage{ragged2e}

\begin{document}



\title{Impact of Fading Correlation on the High-SNR Regime of Reconfigurable Intelligent Surfaces}


\author{
Paula Isabel Tilleria Lucero, Bryan Fernando Sarango Rodríguez, Fernando Darío Almeida García, \\ and José Cândido Silveira Santos Filho 
\thanks{This work was financed in part by the Coordenação de Aperfeiçoamento de Pessoal de Nível Superior - Brasil (CAPES) - Finance Code 001.
This work was partially funded by the Brasil 6G Project with support from RNP/MCTI (Grant 01245.010604/2020-14), and the xGMobile Project code XGM-FCRH-2024-10-16-1 with resources from EMBRAPII/MCTI (Grant 052/2023 PPI IoT/Manufatura 4.0 and FAPEMIG Grant PPE-00124-23).}
\thanks{P. I. T. Lucero, B. F. S. Rodríguez, and J. C. S. Santos Filho are with the Wireless Technology Laboratory, Department of Communications, School of Electrical and Computer Engineering, State University of Campinas (UNICAMP), Campinas, SP  13083-852, Brazil (e-mail: \mbox{p249999@dac.unicamp.br}; \mbox{b241852@dac.unicamp.br}; \mbox{candido@decom.fee.unicamp.br}).
F. D. A. García is with the Wireless and Artificial Intelligence (WAI) laboratory, National Institute of Telecommunications (INATEL), Santa Rita do Sapucaí, MG, 37540-000, Brazil (e-mail: \mbox{fernando.almeida@inatel.br}).}
}

\maketitle

\begin{abstract} 
This paper addresses three critical limitations in previous analyses of RIS-aided wireless systems: propagation environments with fixed diversity gain, restricted spatial correlation profiles, and approximation methods that fail to capture the system behavior in the high signal-to-noise ratio (SNR) regime.
To overcome these challenges, we conduct an exact asymptotic analysis focused on the left tail of the SNR distribution, which plays a critical role in high-SNR system performance. Additionally, to account for general correlation profiles and fading environments with variable diversity and coding gains, we consider arbitrarily correlated Nakagami-$m$ fading channels.
The analytical results show that fading correlation induces a horizontal shift in the asymptotic behavior---represented as a straight line in the log-dB scale---of the PDF and CDF, displacing these curves to the left. The asymptotic linear coefficient quantifies this shift, while the angular coefficient remains unaffected. Moreover, the results reveal that the high sensitivity of the linear coefficient to correlation arises from the aggregated contribution of all marginal asymptotic terms, effectively capturing each channel's correlation characteristics.
\end{abstract}

\section{Introduction}

Reconfigurable intelligent surfaces (RIS) offer a transformative paradigm to wireless communications by enhancing signal quality and network performance~\cite{Basar19}. 
However, RIS-assisted wireless networks still face key challenges, particularly in channel estimation, beamforming, sum-rate optimization, and channel modeling. In this work, we focus on the latter: channel modeling.
Specifically, we aim to derive the chief statistics of the system's signal-to-noise ratio (SNR), namely the probability density function (PDF) and cumulative distribution function (CDF).
These metrics are essential for evaluating system performance and guiding the design of RIS-enabled~networks.


Several studies have investigated RIS channel modeling and derived the PDF and CDF of the SNR under the assumptions of independent and identically distributed (i.i.d.) or independent non-identically distributed (i.n.i.d.) fading channels (cf.~\cite{Badiu20,Selimis21,Almeida2024,Trigui22,Amiriara24,Luna24,Agbogla23,Singh23} and references therein).
However, in practical deployments, mutual correlation among RIS elements---arising from their close physical spacing---significantly affects signal reception~\cite{Björnson2021,Wang23FD}. Therefore, accurate performance evaluation requires incorporating spatial correlation into the channel modeling process.

While independent fading scenarios have been extensively studied, only a few works have focused on channel modeling for RIS-aided systems under correlated fading conditions. Notable contributions include those in~\cite{Björnson2021, Trigui2024,Wang23FD, Papazafeiropoulos21,Chien21}.
In~\cite{Björnson2021}, spatial correlation due to the close proximity of RIS elements was modeled under Rayleigh fading, and the SNR statistics (PDF and CDF) were approximated in the large-RIS regime (i.e., a large number of RIS elements)  using the Central Limit Theorem (CLT).
In~\cite{Papazafeiropoulos21}, arbitrarily correlated Rayleigh fading was considered, and the SNR statistics were approximated in the large-RIS regime using deterministic equivalent (DE) analysis.
In~\cite{Wang23FD}, an exponentially decaying Nakagami-$m$ correlation was considered, while the SNR statistics were approximated by a Gamma distribution using moment-matching (MM) technique. Nonetheless, both the MM technique and the CLT often fail to accurately capture the tail behavior of the distribution. 
In~\cite{Chien21}, arbitrary spatial Rayleigh correlation was considered, and the SNR statistics were approximated using a Gamma distribution via the MM technique.
More recently,~\cite{Trigui2024} presented an exact and asymptotic analysis of SNR statistics under equally correlated Rician fading. The asymptotic expressions were derived using residue calculus, while the exact formulations were expressed in terms of nested series involving generalized multivariate Fox’s H-functions. Although thorough, the multivariate Fox’s H-function approach entails significant computational complexity as the number of RIS elements increases~\cite{Garcia22alpha_mu}.

Despite considerable efforts, several limitations remain in the analysis of RIS-assisted networks over correlated fading channels: (i) propagation environments restricted to
fixed diversity gain (e.g., Rician fading \cite{Trigui2024});
(ii) simplified correlation profiles (e.g., equally and exponentially decaying \cite{Trigui2024,Wang23FD}); and (iii) approximation methods that inadequately describe the tail behavior of the SNR distribution (e.g., MM and CLT approaches \cite{Wang23FD,Björnson2021,Chien21}).
In this paper, we address these limitations through an exact asymptotic analysis focused on the left tail of the distribution---a region critical for characterizing system behavior in the high-SNR regime.
Moreover, to extend the analysis to unrestricted correlation structures and propagation environments with variable coding gain, we consider arbitrarily correlated Nakagami-$m$ fading~channels.

In the sequel, the operator $\mathbb{E}(\cdot)$ denotes expectation; $\mathbb{V}(\cdot)$, variance; $\left[\cdot\right]^T$ indicates a transposed vector; $\left(\cdot\right)^*$, complex conjugate; $\mathbb{R}^N_+$, the set of all $N$-dimensional real vectors with non-negative components; $\mathbb{C}$, the set of complex numbers; $\Re\left\{\cdot\right\}$ and $\Im\left\{\cdot\right\}$ denote the real and imaginary parts, respectively, of a complex number; $\mathbbm{i}=\sqrt{-1}$ is the imaginary unit;  $\text{diag} \left( \cdot \right)$, returns the diagonal elements of a square matrix; $I_{\nu}\left(\cdot\right)$, the modified Bessel function of the first kind of $\nu$th order; $K_{\nu}(\cdot)$, the modified Bessel function of the second kind of $\nu$th order; $\Gamma(\cdot)$, the gamma function; ${}_{2}\tilde{F}_{1} (\cdot,\cdot,\cdot,\cdot)$, the regularized hypergeometric function; $\circ$ and $\oslash$, the Hadamard (element-wise) product and division, respectively; $\Pi (\mathbf{A})$, the product of all elements of the vector $\mathbf{A}$;  $\arg\max_{a \in \mathcal{A}} g_a$ returns the argument $a$ at which $g_a$ attains its maximum; and $\sim$, ``asymptotically equal to around zero, '' i.e., ${h(x)\sim g(x) \iff \lim\limits_{x\to 0} \, \frac{h (x)}{g (x)}=1}$.




\section{System and Channel Modeling}
\label{sec:SystemModel}
In this section, we briefly review the definition and construction of correlated Nakagami-$m$ random variables (RVs). Then, we introduce the adopted RIS-assisted wireless system.




\subsection{Correlated Nakagami-$m$ Random Variables}
\label{sec:CorrelatedNakagami}
The correlation among Nakagami-$m$ RVs is induced by their underlying Gaussian components.
Let $\left\{ G_{n,l}\right\}_{n=1,l=1}^{N,m}$ denote a set of complex Gaussian RVs constructed as in \cite[eq. (7)]{Beaulieu2011}:
\par\nobreak\vspace{-\abovedisplayskip}
\small
\begin{align}
    G_{n,l} &= \sigma_n \left( \sqrt{1 - \lambda_n^2} X_{n,l} + \lambda_n X_{0,l} \right) \nonumber \\
    & + \mathbbm{i} \sigma_n \left( \sqrt{1 - \lambda_n^2} Y_{n,l} + \lambda_n Y_{0,l} \right),
\end{align}
\normalsize
where $X_{0,l}, Y_{0,l}$, $X_{n,l}$, and $Y_{n,l}$  ($ n = 1,\hdots,N $ and $ l = 1,\hdots,m $) are independent Gaussian RVs with zero mean and variance $1/2$.
Then, for any $ n, k \in \{1,\dots,N\} $, and $ l, j \in \{1,\dots,m\} $, we have $\mathbb{E}(X_{n,l} Y_{k,j}) = 0$, and $\mathbb{E}(X_{n,l} X_{k,j}) =\mathbb{E}(Y_{n,l} Y_{k,j}) =\frac{1}{2} \varrho_{n,k} \varrho_{l,j} $, where $\varrho_{a,b}$ is the Kronecker delta function~\cite[eq. (04.20.02.0001.01)]{wolframFunctions}.
The cross-correlation coefficient between any pair of complex Gaussian RVs $G_{n,l}$ and $G_{k,j}$, $ \forall n \ne k$, is given by \cite[eq. (8)]{Beaulieu2011}
\par\nobreak\vspace{-\abovedisplayskip}
\small
\begin{align}
    \bar{\rho}_{(n,l),(k,j)} &= \frac{\mathbb{E}(G_{n,l} G_{k,j}^*) - \mathbb{E}(G_{n,l}) \mathbb{E}(G_{k,j}^*)} {\sqrt{ \mathbb{E}(|G_{n,l}|^2) \mathbb{E}(|G_{k,j}|^2) }} \nonumber \\ 
    & = \begin{cases} 
        1, & n = k \text{ and } l = j \\
        \lambda_n \lambda_k, & n \neq k \text{ and } l = j \\
        0, & l \neq j. 
    \end{cases}
\end{align}
\normalsize
Moreover, $\lambda_n \in [-1,1] \backslash \left\{0\right\} $ is a correlation parameter that governs the statistical dependence between any $G_{n,l}$ and $G_{k,j}$, and $\sigma_n^2 = \mathbb{E} (\Re\left\{G_{n,l}\right\}^2) = \mathbb{E} (\Im\left\{G_{n,l}\right\}^2)$.
Each resulting RV $H_n = \sqrt{\sum_{l=1}^{m} \left|{G_{n,l}}\right|^2}$
follows a Nakagami-$m$ distribution with shape \mbox{parameter} $m$ and mean square value  $\mathbb{E}(H_{n}^2) = m \sigma_{n}^2$.
The cross-correlation coefficient between any pair of Nakagami-$m$ RVs $H_n$ and $H_k$ is given by \cite[eq. (10)]{Beaulieu2011}
\par\nobreak\vspace{-\abovedisplayskip}
\small
\begin{align}
    \rho_{n,k} &= \frac{\mathbb{E}(H_n H_k^*) - \mathbb{E}(H_n) \mathbb{E}(H_k^*)} {\sqrt{ \mathbb{V}(H_n) \mathbb{V}(H_k) }} = \lambda_n^2 \lambda_k^2.
\end{align}
\normalsize


Let $\mathbf{H} \triangleq [H_1, H_2, ..., H_N]^T$ denote a random vector consisting of correlated Nakagami-$m$ RVs.
The joint PDF of $\mathbf{H}$ is given by \cite[eq. (20)]{Beaulieu2011}
\par\nobreak\vspace{-\abovedisplayskip}
\small
\begin{align}
    f_{\mathbf{H}}(\mathbf{h}) &= \int_0^{\infty} \frac{t^{m-1}}{\Gamma(m)} \exp(-t) \prod_{n=1}^{N} \frac{1}{(\sigma_{n}^2\lambda_{n}^2 t)^{\frac{m-1}{2}}} \frac{h_{n}^m}{\delta_{n}^2} \nonumber \\
    & \times \exp\left(-\frac{h_{n}^2+\sigma_{n}^2 \lambda_{n}^2 t}{2 \, \delta_{n}^2}\right)  I_{m-1}\left(\frac{h_{n}\sqrt{\sigma_{n}^2\lambda_{n}^2 t}}{\delta_{n}^2}\right) \text{d}t,
    \label{eq:JointCorreNakaPDF}
\end{align}
\normalsize
where $\mathbf{h} = [h_1, h_2, \hdots, h_N]^T$,  $\delta_{n}^2 = \sigma_{n}^2\left(\frac{1-\lambda_{n}^2}{2}\right)$, and  $m \geq 0.5$ denotes (as before) the fading shape parameter.



\subsection{RIS-aided Wireless Communications}
\label{sec:RISNakagami}

Consider a wireless network in which a RIS comprising $N$ elements assists communication between a single-antenna user transmitter (UT) and a user receiver (UR). The UT--RIS and RIS--UR channels experience correlated Nakagami-$m$ fading, while the direct UT--UR link is obstructed by large obstacles.

Let $\tilde{\mathbf{H}}_{1} \triangleq [ \tilde{H}_{1,1}, \dots, \tilde{H}_{1,N} ]^T \in \mathbb{C}^{N \times 1}$ and $\tilde{\mathbf{H}}_2 \triangleq  [ \tilde{H}_{2,1}, \dots, \tilde{H}_{2,N} ]^T \in \mathbb{C}^{N \times 1}$ denote the complex random vectors containing the channel coefficients between the UT and the $n$-th RIS element, and between the $n$-th RIS element and the UR, respectively.
Each channel coefficient is expressed as $\tilde{H}_{v,n}= H_{v,n} e^{\mathbbm{i} \phi_{v,n}}$ ($v \in {1,2}$), where $H_{v,n}$ and $\phi_{v,n}$ denote the envelope and phase components, respectively.
Moreover, define $\mathbf{H}_1 \triangleq [H_{1,n}, \hdots,H_{1,N}]^T$ and $\mathbf{H}_2 \triangleq [H_{2,n}, \hdots,H_{2,N}]^T$.
In this paper, we assume that $\mathbf{H}_1$ and $\mathbf{H}_2$ are mutually independent random vectors, each consisting of $N$ correlated Nakagami-$m$ RVS with joint PDF given by~\eqref{eq:JointCorreNakaPDF}.
Additionally, we consider that the received signal at the UR is influenced by the RIS reflection matrix, which controls the phase shifts applied to the incident signal. The reflection matrix is defined as $\Theta = \operatorname{diag} \left( e^{\mathbbm{i} \theta_1}, \hdots,  e^{\mathbbm{i} \theta_N} \right)$, where $\theta_n \in [0, 2\pi]$ denotes the phase shift introduced by the $n$-th RIS element.
Lastly, we consider that the RIS applies optimal phase shifts to maximize the SNR.
Under these considerations, the instantaneous SNR at the UR is expressed as
\par\nobreak\vspace{-\abovedisplayskip}
\small
\begin{align}
    \gamma = \bar{\gamma}\left|\sum_{n=1}^{N}H_{1,n} H_{2,n}\right|^2,
    \label{eq:SNR}
\end{align}
\normalsize
where $\bar{\gamma}$ represents the average SNR per symbol.

Eq.~\eqref{eq:SNR} provides the basis for performance analysis of the RIS-assisted system under arbitrarily correlated Nakagami-$m$ fading. However, it is worth noting that deriving the statistical characterization of the SNR statistics (PDF and CDF) is a highly challenging task. To date, no tractable expressions have been obtained for the PDF or CDF of~\eqref{eq:SNR}.
To address this analytical challenge, we adopt an asymptotic approach focused on the left tail of the PDF of $\gamma$, near the origin.
As shown next, the asymptotic analysis yields simplified expressions for the SNR statistics, enabling efficient performance evaluation and design insights at high SNR.





\section{Asymptotic SNR Analysis}
\label{sec:AsymptoticAnalysis}



In this section, we perform an asymptotic analysis aiming at deriving asymptotic expressions for the PDF and CDF of~$\gamma$.

Let the sum in \eqref{eq:SNR} be denoted as
\par\nobreak\vspace{-\abovedisplayskip}
\small
\begin{align}
    S \triangleq \sum_{n=1}^{N}H_{1,n} H_{2,n}.
    \label{eq:RisModel}
\end{align}
\normalsize



Moreover, let the PDF of $S$ be expressed via its Maclaurin series expansion as $f_S (s) = \sum_{n=0}^{\infty} a_n s^{b_n}$,
where $a_n$ and $b_n$ are the expansion coefficients.
According to \cite[eq. (3)]{Parente2019}, the PDF of $S$ admits the asymptotic representation
\par\nobreak\vspace{-\abovedisplayskip}
\small
\begin{align}
    f_S(s) \sim a_0 s^{b_0},
    \label{eq:AsympSum}
\end{align}
\normalsize
where $a_0$ and $b_0$ are the asymptotic linear and angular coefficients (in a log-scale visualization) of the PDF of $S$, respectively.\footnote{The linear and angular coefficients are crucial in asymptotic analysis as they influence the coding and diversity gains of diverse performance metrics, such as the outage probability or average symbol error probability~\cite{Parente2019}.}
The asymptotic representation of the PDF of $\gamma$ can be readily derived from \eqref{eq:AsympSum} via a standard change of variables, $\gamma = \bar{\gamma}|S|^2$, yielding 
\par\nobreak\vspace{-\abovedisplayskip}
\small
\begin{align}
    f_{\gamma}(\gamma) \sim a_{\gamma,0} \gamma^{b_{\gamma,0}},
    \label{eq:PDFSNR}
\end{align}
\normalsize
where  $a_{\gamma,0}= a_0 / \left(2 \, \bar{\gamma}^{\frac{b_0 + 1}{2}}\right)$ and $b_{\gamma,0}= (b_{0}-1)/2$ denote the linear and angular coefficients the PDF of $\gamma$, respectively.
The CDF of $\gamma$ can be obtained by integrating \eqref{eq:PDFSNR}  from zero to $\gamma$, i.e., $\int_{0}^{\gamma} f_{\gamma}(u) \text{d}u$, resulting in
\par\nobreak\vspace{-\abovedisplayskip}
\small
\begin{align}
     F_{\gamma}(\gamma) \sim \frac{a_{\gamma,0}}{1 + b_{\gamma,0}} \gamma^{1 + b_{\gamma,0}}.
        \label{eq:CDFSNR}
\end{align}
\normalsize
Accordingly, our main goal is to determine the linear and angular coefficients of the PDF of $S$, namely $a_0$ and $b_0$.  
This is attained in the next proposition.



\begin{proposition}
Given the asymptotic representation of the PDF of $S$ in \eqref{eq:AsympSum}, their corresponding linear and angular coefficients are given, respectively, by
\par\nobreak\vspace{-\abovedisplayskip}
\small
\begin{subequations}
\begin{align}
    a_{0} = &
    \frac{\Phi \, \Gamma(2 \, m_{\text{<}})^N  }
    { \Gamma(2 N m_{\text{<}})} \sum_{i_{1},\dots,i_{N}=0}^{\infty} \frac{\Gamma\left( m_{\text{>}} + \sum_{l=1}^{N} i_l \right)}{\left(2 + \sum_{l=1}^{N} \frac{\lambda_{r,l}^2 \sigma_{r,l}^2}{\delta_{r,l}^2} \right)^{\sum_{l=1}^{N} i_l}} \nonumber \\
    & \times   \prod_{n=1}^{N} \frac{\Gamma(i_n + |m_{\text{--}}|)  } {\Gamma(i_n + m_{\text{>}}) i_n!} \left( \frac{\sigma_{r,n} \lambda_{r,n}}{\delta_{r,n}} \right)^{2 i_n} 
    \label{eq:a0} \\
    b_{0} = & 2 N m_{\text{<}}- 1,
    \label{eq:b0}
\end{align}
\label{eq:asymTerms}
\end{subequations}
\normalsize
where $m_{\text{+}} \triangleq m_1+m_2$,  $m_{\text{--}} \triangleq  m_1-m_2$, $m_{\text{<}} \triangleq \min (m_1,m_2)$, $m_{\text{>}} \triangleq \max (m_1,m_2)$, $r \triangleq \arg\max_{v \in \{ 1, 2 \}} m_v$, and 
\par\nobreak\vspace{-\abovedisplayskip}
\small
\begin{align}
    \label{eq: phi}
    \Phi = \frac{\frac{2^{N(1 - 2 \, m_{\text{<}})} }{\Gamma(m_{\text{>}}) \, \Gamma(m_{\text{<}})^N }\left( \prod_{n=1}^{N} \frac{1}{\delta_{1,n}^2  \delta_{2,n}^2} \right)^{m_{\text{<}}}}{ \left( 1 + \sum_{l=1}^{N} \frac{\lambda_{1,l}^2 \sigma_{1,l}^2}{2 \delta_{1,l}^2} \right)^{m_1} \left( 1 + \sum_{l=1}^{N} \frac{\lambda_{2,l}^2 \sigma_{2,l}^2}{2 \delta_{2,l}^2} \right)^{m_2}}.
\end{align}
\normalsize

\end{proposition}
\begin{proof}
Please, see the Appendix. 
\end{proof}


Leveraging \eqref{eq:a0} and \eqref{eq:b0}, the linear and angular coefficients in \eqref{eq:PDFSNR} can finally be expressed, respectively, as
\par\nobreak\vspace{-\abovedisplayskip}
\small
\begin{subequations}
\begin{align}
    a_{\gamma,0} =&
    \frac{\Phi \, \Gamma(2 \, m_{\text{<}})^N  }
    { 2 \, \Gamma(2 N m_{\text{<}})  \bar{\gamma}^{\frac{b_{0} +1 }{2}}} \sum_{i_{1},\hdots,i_{N}=0}^{\infty} \frac{\Gamma\left( m_{\text{>}} + \sum_{l=1}^{N} i_l \right)}{\left(2 + \sum_{l=1}^{N} \frac{\lambda_{r,l}^2 \sigma_{r,l}^2}{\delta_{r,l}^2} \right)^{\sum_{l=1}^{N} i_l}} \nonumber \\
    &\times \prod_{n=1}^{N} \frac{\Gamma(i_n + |m_{\text{--}}|) } {\Gamma(i_n + m_{\text{>}}) i_n!} \left( \frac{\sigma_{r,n} \lambda_{r,n}}{\delta_{r,n}} \right)^{2 i_n}  
    \label{eq:a0SNR} \\
    b_{\gamma,0} =& N m_{\text{<}}-1.
    \label{eq:b0SNR} 
\end{align}
\label{eq:asymTermsSNR}
\end{subequations}
\normalsize  
It is worth noting that \eqref{eq:PDFSNR}, \eqref{eq:CDFSNR}, and \eqref{eq:asymTermsSNR} form the analytical foundation for evaluating system performance in the high-SNR regime.
In particular, from \eqref{eq:asymTermsSNR}, it is evident that $a_{\gamma,0}$ aggregates contributions from all marginal asymptotic terms, thereby encapsulating the channel correlation parameters ($\lambda_n$) of each individual channel.
In contrast, the angular coefficient $b_{\gamma,0}$ remains unaffected by channel correlation.
Importantly, the $N$-fold series in \eqref{eq:a0SNR} cannot be asymptotically reduced, i.e., retaining only the first, supposedly dominant, terms of each series fail to capture the true asymptotic behavior.
This is because the higher-order terms become increasingly significant as the correlation coefficients increase. Therefore, all terms in the $N$-fold series must be included.
As will be shown in \ref{sec:NumericalResults}, correlation---reflected through $a_{\gamma,0}$---primarily induces a horizontal shift in the asymptotic behavior of the PDF or CDF of $\gamma$, shifting it to the left. Notably, the slope remains unchanged at $N m_{\text{<}}-1$, consistent with the case of independent fading. 
Thus, $a_{\gamma,0}$ quantifies the horizontal shift in the asymptotic behavior caused by the presence of correlated fading channels.
The behavior and performance implications of $a_{\gamma,0}$ will be further explored in Section~\ref{sec:NumericalResults}.


Taking the limit as  $\lambda_{r,n} \to 0$, the linear coefficient for the independent case can be obtained as
\par\nobreak\vspace{-\abovedisplayskip}
\small
\begin{align}
    a_{\gamma_{\text{ind}},0} &= 
    \frac{2^{N-1}\left(\frac{\Gamma\left( |m_{-}| \right) \, \Gamma\left( 2 \, m_{<} \right)}
    {\Gamma(m_1) \, \Gamma(m_2)}\right)^{N}}{\Gamma \left( 2 N m_{<} \right) \bar{\gamma}^{\frac{b_{0} +1 }{2}}} \prod_{n=1}^{N} \left( \frac{1}{ \sigma_{1,n}^2 \, \sigma_{2,n}^2} \right)^{m_{<}}. 
    \label{eq:a0SNR_Independent}
\end{align}
\normalsize
The angular coefficients for the correlated and independent cases remain the same, i.e., $ b_{\gamma_{\text{ind}},0}=b_{\gamma,0} = N m_{\text{<}}-1$.

\subsubsection{Truncation Error}
\label{sec:TruncationError}

Note that \eqref{eq:a0} involves $N$ nested series. 
In practice, however, the computation of the linear coefficient in~\eqref{eq:a0} requires truncating the infinite series to a finite number of terms, denoted here by $\tau$.
This truncation process gives rise to an approximation error, known as truncation error, which from~\eqref{eq:a0} can be obtained as
\par\nobreak\vspace{-\abovedisplayskip}
\small
\begin{align}
    \epsilon_{a_0} (\tau) &=
    \frac{\Phi \, \Gamma(2 \, m_{\text{<}})^N  }
    { \Gamma(2 N m_{\text{<}})} \sum_{i_{1},\dots,i_{N}=\tau}^{\infty} \frac{\Gamma\left( m_{\text{>}} + \sum_{l=1}^{N} i_l \right)}{\left(2 + \sum_{l=1}^{N} \frac{\lambda_{r,l}^2 \sigma_{r,l}^2}{\delta_{r,l}^2} \right)^{\sum_{l=1}^{N} i_l}} \nonumber \\
    & \times   \prod_{n=1}^{N} \frac{\Gamma(i_n + |m_{\text{--}}|)  } {\Gamma(i_n + m_{\text{>}}) i_n!} \left( \frac{\sigma_{r,n} \lambda_{r,n}}{\delta_{r,n}} \right)^{2 i_n}. 
    \label{eq: truncation error def} 
\end{align}
\normalsize

Since the $N$-fold series in \eqref{eq: truncation error def} lacks a closed-form solution, deriving an upper bound for $\epsilon_{a_0} (\tau) $ offers a practical and effective alternative. To approach this, we focus on the term $\mathcal{I} \triangleq \Gamma(m_{\text{>}} + \sum_{l=1}^{N} i_l)$. 
The key idea is to establish an upper bound for $\mathcal{I}$ that decouples the composite argument $m_{\text{>}} + \sum_{l=1}^{N} i_l$ from the gamma function, thereby allowing the 
$N$-fold series in \eqref{eq: truncation error def} to be evaluated in closed form.
Leveraging the log-convexity of the gamma function, we have $\mathcal{I} \leq \prod_{n=1}^{N} \Gamma( m_{\text{>}}+ i_n) K$,
where $K$ is a correction factor that accounts for the gamma function’s submultiplicative behavior over sums. 
Considering the growth behavior of the gamma function and the slack introduced by the exponential function, we set $K=\exp(i_n)$. 
This yields the upper bound $\mathcal{I}  \leq \prod_{n=1}^{N} \Gamma(m_{\text{>}} + i_n) \exp(i_n)$, where the exponential term serves to relax the bound, ensuring that the inequality remains valid for all $i_n$ and $m_{\text{>}}$.
Substituting this bound into \eqref{eq: truncation error def} and after lengthy algebraic manipulations with the aid of~\cite[eq. (07.24.02.0001.01)]{wolframFunctions}, we get
\par\nobreak\vspace{-\abovedisplayskip}
\small
\begin{align}
    \epsilon_{a_0} (\tau) &<
    \frac{\Phi \Gamma(2 \, m_{\text{<}})^{N}}
    { \Gamma(2 N m_{\text{<}})} \prod_{n=1}^{N} \frac{  e^{\tau} \, \Gamma\left( |m_{\text{--}}| + \tau \right)}{\left( 2 + \sum_{l=1}^{N} \frac{\lambda_{r,l}^2 \sigma_{r,l}^2}{\delta_{r,l}^2} \right)^{\tau}} \left( \frac{\sigma_{r,n}\lambda_{r,n}}{\delta_{r,n}} \right)^{2\tau} \nonumber \\ 
    &\times { }_2\tilde{F}_1 \left(1, |m_{\text{--}}| + \tau, 1 + \tau, \frac{e \left( \frac{\sigma_{r,n} \lambda_{r,n}}{\delta_{r,n}} \right)^2 }{2 + \sum_{l=1}^{N} \frac{\lambda_{r,l}^2\sigma_{r,l}^2}{\delta_{r,l}^2}} \right).
    \label{eq:truncationErrorClosedForm}
\end{align}
\normalsize

In practice, the truncation bound in \eqref{eq:truncationErrorClosedForm} plays a pivotal role, as it establishes the sufficient number of terms required in \eqref{eq:a0} to ensure a predefined accuracy (error) threshold. 


\section{Numerical Results}
\label{sec:NumericalResults}


In this section, we validate the proposed analytical formulations via Monte Carlo simulations and investigate the effects of correlated fading.
For the analysis, we adopt the following parameters: $\sigma_{1,n}=\sigma_{2,n}=1$ and  $\lambda \triangleq \lambda_{1,n} = \lambda_{2,n}$ for all $n$, i.e., we assume uniform correlation among the elements of both $\mathbf{H}_1$ and $\mathbf{H}_2$ for illustrative purposes---although the proposed analysis remains applicable to scenarios with arbitrary correlation.
Consequently, the correlation coefficient is $\rho  \triangleq \rho_{n,k} = \lambda^4$ for all  $n \neq k$. 
Finally, we use $\tau=20$ terms in the linear coefficient for all scenarios, since the truncation error in \eqref{eq:truncationErrorClosedForm} decreases rapidly with increasing number of RIS elements $N$, reaching values below $10^{-50}$.

\begin{figure}[t!] 
	\centering
	\includegraphics[scale=0.37]{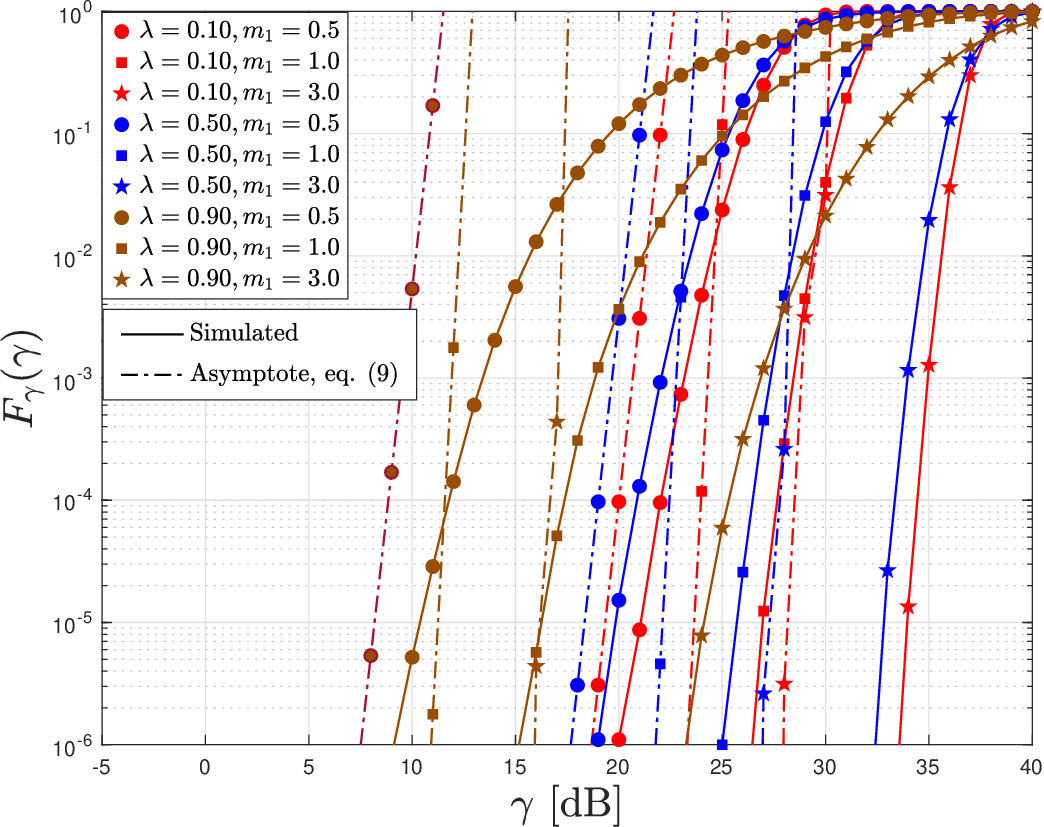}
	\vspace{-0.1cm}
	\caption{Exact and asymptotic CDF of $\gamma$ for $m_2=2$, $N=30$, $\bar{\gamma}=1\text{ dB}$, and different values of $\lambda$ and $m_1$.
    }
	\label{fig:CDFSNR_CorrVariable}
\end{figure}


\begin{figure}[t!]
    \centering
    \begin{subfigure}[t]{0.44\textwidth}
        \centering
        \includegraphics[trim={0.0cm 0.7cm 0cm 0cm},clip,scale=0.34]{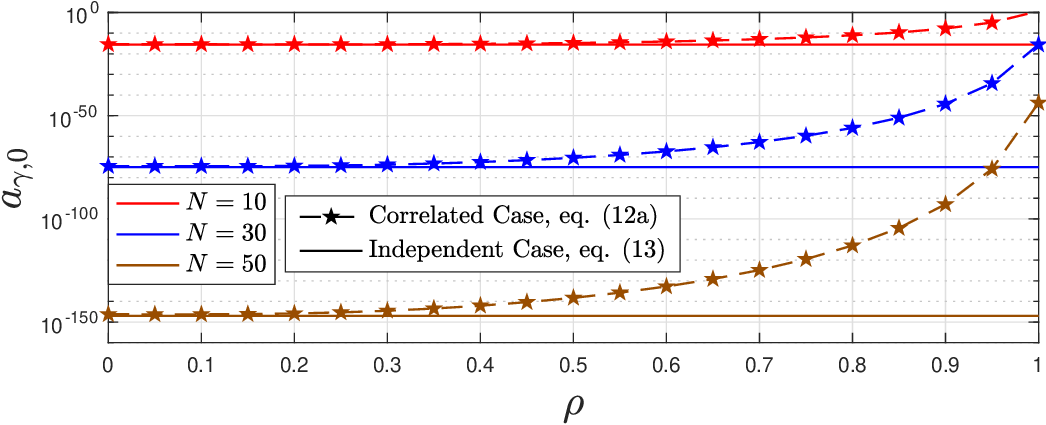}
        \caption{\footnotesize $m_1=2$, $m_2=1$, and different values of $N$.}
        \label{fig:a0SNR_NVariable}
    \end{subfigure}
    \begin{subfigure}[t]{0.44\textwidth}
        \centering
        \includegraphics[scale=0.34]{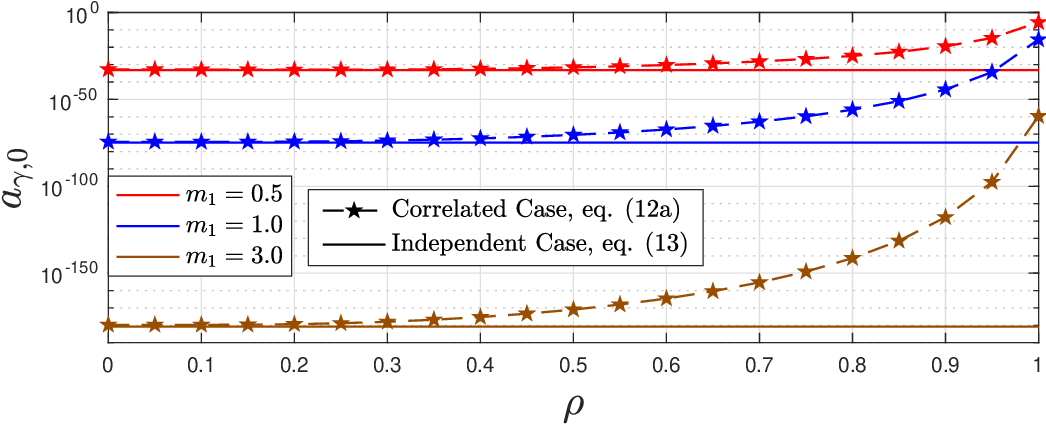}
        \caption{\footnotesize $m_2=2$, $N=30$, and different values of $m_1$.}
        \label{fig:a0SNR_mVariable}
    \end{subfigure}
    \vspace{-0.2cm}
    \caption{Asymptotic linear coefficient $a_{\gamma,0}$ versus correlation coefficient $\rho$ for $\bar{\gamma}=1\text{ dB}$ and different values of $N$ and $m_1$.}
    \label{fig:a0SNR}
\end{figure}


Fig. \ref{fig:CDFSNR_CorrVariable} depicts the exact and asymptotic CDF of $\gamma$ considering $m_2=2$, $N=30$, $\bar{\gamma}=1\text{ dB}$, and different values of $\lambda$ and $m_1$.
Notice that, for a fixed $m_1$, both the exact and asymptotic curves shift progressively to the left as $\lambda$ increases. A similar trend is observed when $\lambda$ is fixed and $m_1$ decreases. This leftward shift indicates that a stronger correlation or more severe fading (i.e., lower $m_1$) reduces the coding gain at a given average SNR per symbol, degrading system performance.
The worst-case scenario occurs when $m_1$ reaches its minimum value and $\lambda$ its maximum, reflecting the combined impact of deep fading and strong correlation.

Fig. \ref{fig:a0SNR} illustrates the linear coefficient $a_{\gamma,0}$ (for both independent and correlated cases) as a function of the correlation coefficient $\rho$ assuming $\bar{\gamma}=1\text{ dB}$ and different values of $N$ and $m_1$.
As shown in Fig.~\ref{fig:a0SNR}-(b), increasing the number of RIS elements $N$ or the fading shape parameter $m_1$ (larger values of $m_1$ correspond to less severe fading conditions) leads to a reduction in $a_{\gamma,0}$.
This trend is consistent in both the independent and correlated cases. Intuitively, this corresponds to a rightward shift of the exact and asymptotic CDF curves, resulting in higher coding gain at a given average SNR per
symbol, thus improving system performance.
Also, from Figs.~\ref{fig:a0SNR}-(a) and \ref{fig:a0SNR}-(b), it is observed that for the correlated case, $a_{\gamma,0}$ remains nearly constant for values of $\rho \lesssim 0.5$, aligning with the value observed in the independent case---as expected. However, for $\rho \gtrsim 0.5$, $a_{\gamma,0}$ rapidly increases and reaches its maximum at $\rho = 1$.
At first glance, the extremely small variations in $a_{\gamma,0}$, on the order of $10^{-100}$, might seem negligible for system performance. However, this is not the case, as it is the combined effect of both $a_{\gamma,0}$ and $b_{\gamma,0}$ that determines the system's coding gain~\cite{Parente2019}. Consequently, even slight changes in $a_{\gamma,0}$ can lead to significant improvements or degradations in overall system performance.

\section{Conclusions}

This study analyzed the SNR statistics of RIS-assisted wireless communication systems operating over arbitrarily correlated Nakagami-$m$ fading channels. The asymptotic analysis demonstrated that spatial correlation primarily influences the linear coefficient, while leaving the angular coefficient unchanged. It also quantified how stronger correlation and more severe fading degrade system performance.



\section*{Appendix \\Proof of the Proposition}



Let $\mathbf{Z}\triangleq [Z_1,\hdots,Z_N]^T= \mathbf{H}_1 \circ \mathbf{H}_2$.
This product can be obtained via the $N$-dimensional Mellin convolution as
\par\nobreak\vspace{-\abovedisplayskip}
\small
\begin{align}
    \label{eq: N fold integral}
    f_{\mathbf{Z}} (\mathbf{z}) =  \int_{\mathbb{R}^N_+}  f_{\mathbf{H}_1} \left(\mathbf{h}_1 \right) f_{\mathbf{H}_2} \left(\mathbf{z} \oslash \mathbf{h}_1\right) \frac{1}{\Pi (\mathbf{h}_1) }\text{d} \mathbf{h}_1,
\end{align}
\normalsize
where $\mathbf{z}= [z_1,\hdots,z_N]^T$, $\mathbf{h}_1= [h_{1,1},\hdots,h_{1,N}]^T$, and $\mathbf{h}_2= [h_{2,1},\hdots,h_{2,N}]^T$. 
Due to the interdependence among the elements of $\mathbf{H}_1$ and $\mathbf{H}_2$, the elements of $\mathbf{Z}$ are inherently~correlated.
Since the joint PDFs of $\mathbf{H}_1$ and $\mathbf{H}_2$, both given by \eqref{eq:JointCorreNakaPDF}, are expressed in terms of a single integral involving products of modified Bessel functions of the first kind, the $N$-fold integral in \eqref{eq: N fold integral} cannot be evaluated in closed form.
To address this, we substitute the modified Bessel functions of the first kind with their series representations~\cite[eq. (03.02.02.0001.01)]{wolframFunctions}.
Then, by interchanging the order of summation (from the Bessel series) and integration, the resulting $N$-fold integral becomes tractable and can be computed term by term.
This yields a new expression for the PDF of $\boldsymbol{Z}$, now given in terms of modified Bessel functions of the second kind, as follows:
\par\nobreak\vspace{-\abovedisplayskip}
\small
\begin{align}
f_{\boldsymbol{Z}}(\boldsymbol{z}) & = \frac{\left( 1 + \sum_{l=1}^N \frac{\lambda_{1,l}^2 \sigma_{1,l}^2}{2 \, \delta_{1,l}^2} \right)^{-m_1} 
\left( 1 + \sum_{l=1}^N \frac{\lambda_{2,l}^2 \sigma_{2,l}^2}{2 \, \delta_{2,l}^2} \right)^{-m_2}}{\Gamma (m_1) \Gamma (m_2)} \nonumber \\
\times & \left[ \prod_{n=1}^N \left(\frac{1}{2 \, \delta_{1,n}^2}\right)^{m_1-1} 
\left(\frac{1}{2 \, \delta_{2,n}^2}\right)^{m_2-1} 
\delta_{1,n}^{m_{\text{--}} - 2}  \delta_{2,n}^{-m_{\text{--}} - 2} \right] \nonumber \\
\times & \left[\sum_{i_{1},\dots,i_{N}=0}^{\infty} \sum_{j_{1},\dots,j_{N}=0}^{\infty}
\Gamma \left( m_1 + \sum_{l=1}^N i_l \right) \Gamma \left( m_2 + \sum_{l=1}^N j_l \right) \right. \nonumber \\
\times & \prod_{n=1}^N \frac{
    z_n^{m_{\text{+}} + i_n + j_n - 1} 
    K_{m_{\text{--}} + i_n - j_n} \left( \frac{z_n}{\delta_{1,n}  \delta_{2,n}} \right) 
}{  \Gamma \left( m_1 + i_n \right) \Gamma \left( m_2 + j_n \right)   i_n! j_n! } \nonumber \\
\times & \delta_{1,n}^{i_n - j_n} 
    \left( \frac{\lambda_{1,n} \sigma_{1,n}}{2 \, \delta_{1,n}^2} \right)^{2 i_n} 
    \left( \sum_{n=1}^N \frac{\lambda_{1,n}^2 \sigma_{1,n}^2}{2 \, \delta_{1,n}^2} + 1 \right)^{-i_n} \nonumber \\
\times  & \left. \delta_{2,n}^{i_n - j_n} 
    \left( \frac{\lambda_{2,n} \sigma_{2,n}}{2 \, \delta_{2,n}^2} \right)^{2 j_n} 
    \left( \sum_{n=1}^N \frac{\lambda_{2,n}^2 \sigma_{2,n}^2}{2 \,  \delta_{2,n}^2} + 1 \right)^{-j_n} \right].
    \label{eq:PDFProduct}
\end{align}
\normalsize

We now conduct an asymptotic analysis of the PDF of $\boldsymbol{Z}$ near the origin, i.e., in the region where $f_{\boldsymbol{Z}}(\boldsymbol{z}) \to \mathbf{0}$, with $\mathbf{0}$ denoting the $N$-dimensional null vector.
To this end, we first exploit the symmetry property of the modified Bessel function of the second kind, namely, $K_\nu(\cdot) = K_{-\nu}(\cdot)$\cite[eq. (03.04.04.0002.01)]{wolframFunctions}.
Next, we apply the series expansion of the modified Bessel function of the second kind\cite[eq. (6)]{Molu2017}.
Finally, by retaining only the dominant (first) terms in the resulting Bessel series expansion, we obtain the following simplified asymptotic expression:
\par\nobreak\vspace{-\abovedisplayskip}
\small
\begin{align}
    f_{\boldsymbol{Z}}(\boldsymbol{z}) \sim \prod_{n=1}^{N} a_{n,0} z_{n}^{b_{n,0}},
    \label{eq:asympProduct}
\end{align}
\normalsize
where 
\par\nobreak\vspace{-\abovedisplayskip}
\small
\begin{subequations}
\begin{align}
    a_{n,0} = & \Phi^{\frac{1}{N}} \sum_{i_{1},\dots,i_{N}=0}^{\infty} \frac{\Gamma\left( m_{\text{>}} + \sum_{l=1}^{N} i_l \right)^{\frac{1}{N}}}{\left(2 + \sum_{l=1}^{N} \frac{\lambda_{r,l}^2 \sigma_{r,l}^2}{\delta_{r,l}^2} \right)^{\frac{1}{N}\sum_{l=1}^{N} i_l}} \nonumber \\
    & \times \prod_{n=1}^{N} \frac{\Gamma(i_n + |m_{\text{--}}|)^{\frac{1}{N}} } {\left(\Gamma(i_n + m_{\text{>}}) \, i_n!\right)^{\frac{1}{N}}} \left( \frac{\sigma_{r,n} \lambda_{r,n}}{\delta_{r,n}} \right)^{\frac{2 i_n}{N}}
    \label{eq:an0} \\
    b_{n,0} = & 2 \, m_{\text{<}} - 1.
    \label{eq:bn0}
\end{align}
\label{eq: abn0}
\end{subequations}
\normalsize
It can be noticed that \eqref{eq:asympProduct} shares the same functional form as \cite[eq. (5)]{Parente2019}. This structural similarity allows us to infer that the elements of $\mathbf{Z}$ (i.e., $Z_1,\hdots,Z_N$) exhibit asymptotic independence, i.e., in the vicinity of the origin, the joint asymptotic PDF of $\mathbf{Z}$ factorizes into the product of its marginal asymptotic PDFs. With this insight, we are now in a position to derive the asymptotic representation of the sum in \eqref{eq:RisModel}. Specifically, by substituting $a_{n,0}$ and $b_{n,0}$ from \eqref{eq: abn0} into \cite[eq. (4)]{Parente2019}, we obtain the linear and angular coefficients of the PDF of  $S$, as presented in \eqref{eq:asymTerms}. This completes the proof.

\vspace{-1.4cm}

\vspace{1.5cm}
\justifying
\bibliographystyle{IEEEtran}
\bibliography{IEEEabrv,bibliografia}

\end{document}